\documentclass[11pt]{article}

\usepackage[T1]{fontenc}
\usepackage[utf8]{inputenc}
\usepackage{lmodern}
\usepackage[margin=1in]{geometry}
\usepackage{microtype}
\usepackage{amsmath,amssymb,amsthm,mathtools,bm}
\usepackage{booktabs}
\usepackage{graphicx}
\usepackage{subcaption}
\usepackage{enumitem}
\usepackage{algorithm}
\usepackage{algpseudocode}
\usepackage[colorlinks=true,linkcolor=blue,citecolor=blue,urlcolor=blue]{hyperref}
\usepackage{authblk}

\newtheorem{theorem}{Theorem}[section]
\newtheorem{proposition}[theorem]{Proposition}

\theoremstyle{definition}

\newtheorem{definition}[theorem]{Definition}
\theoremstyle{remark}
\newtheorem{remark}[theorem]{Remark}

\newcommand{\R}{\mathbb{R}}

\newcommand{\Prob}{\mathbb{P}}
\newcommand{\dd}{\,\mathrm{d}}

\newcommand{\safeincludegraphics}[2][]{%
  \IfFileExists{#2}{\includegraphics[#1]{#2}}{%
    \fbox{\parbox[c][0.25\textheight][c]{0.92\linewidth}{%
      \centering Figure file to be inserted:\\[0.5em]
      \texttt{\detokenize{#2}}}}}}

\title{Mixed Gaussian Projections for Two-Sample Testing of Functional Data}

\author[1]{Fernando A. Najman}
\author[2]{Marcela Svarc}

\affil[1]{Centro de Matemática, Computação e Cognição,
Universidade Federal do ABC (UFABC), Brazil}
\affil[2]{Departamento de Matemática, Universidad de San Andrés--CONICET, Argentina}

\date{ }

\begin{document}
\maketitle

\begin{abstract}
Random-projection tests for functional data depend on the probability law used to generate projection directions. A measure has to be selected to generate the random directions in which the data is projected. In $L^2$, probability measures defined by their moments can be discriminated using Gaussian measures. The covariance operator of the Gaussian projection law determines which regions and structures of the functional space receive appreciable probability, and consequently affects finite-sample power. We show that mixtures of non-degenerate Gaussian measures preserve the almost-sure separation property and induce a metric between functional probability laws. A concentration argument further makes explicit that the power of projection tests is governed by the integrated projected distance associated with the chosen law. As a concrete construction, we combine Haar and Fourier Gaussian components, which emphasize localized and oscillatory departures, respectively. The resulting permutation test retains consistency, while component labels and upper-tail separation scores provide a descriptive indication of the geometry of the detected discrepancy. Simulations illustrate the specialization of the two components and the robustness of their mixture, and an ECG5000 application identifies a predominantly localized difference between two classes. We also compare with a pooled functional principal component analysis covariance operator. This finite-rank, data-adaptive projection geometry preserves permutation validity through its label-invariant construction and improves power in several simulated settings, particularly under a change in covariance structure.
\end{abstract}

\section{Introduction}

Comparing probability distributions is a fundamental problem in statistics, with applications ranging from goodness-of-fit testing and model validation to classification, clustering, and two-sample inference
\cite{gonzalezmanteiga2023,javanmard2024}.
 In the context of functional data, however, these problems become particularly challenging because observations belong to infinite-dimensional spaces, making classical multivariate procedures difficult to extend. Random projections provide an attractive strategy by reducing comparisons of high-dimensional or functional observations to collections of one-dimensional problems while preserving the ability to discriminate between distinct probability laws.

A key theoretical justification for this approach is provided by the sharp Cramér--Wold theorem of \cite{cuesta2007}. Under a Carleman-type moment condition, a direction drawn from any non-degenerate Gaussian measure separates two distinct functional probability laws almost surely. Consequently, averages of projected Kolmogorov--Smirnov distances define population dissimilarities and lead to consistent empirical procedures \cite{galves2026}.

Beyond goodness-of-fit testing \cite{cuesta2007b,cuesta2007}, random projections have become a general statistical tool for functional data analysis. They have been successfully applied to linear models \cite{cuesta2019}, multiway MANOVA \cite{cuesta2010}, outlier detection and descriptive statistics \cite{fraiman2013}, depth \cite{cuevas2009,cuesta2008} and local depth \cite{fernandezpiana2022} measures, classification \cite{fernandezpiana2022}, and clustering \cite{fernandezpiana2026,galves2026,golovkine2022,mori2026, najman2025}, among many other problems. This broad range of applications highlights the practical importance of understanding the factors that determine the effectiveness of projection-based methods.

The choice of the projection law remains a central practical issue. A centered Gaussian measure on a Hilbert space is uniquely determined by its covariance operator, which determines the geometry of the sampled directions \cite{bosq2000}. Intuitively, a projection law assigning greater probability to directions aligned with the discrepancy between two populations should produce larger projected distances and therefore greater finite-sample power. Conversely, a poorly aligned covariance structure may require substantially more observations or projection directions to achieve the same performance.

This article studies mixtures of Gaussian projection laws. We first show that any mixture supported on non-degenerate Gaussian measures preserves the separation property of \cite{cuesta2007}. This result permits a single projection mechanism to combine covariance structures adapted to different forms of functional variation. We then adapt a concentration argument for empirical projection distances to make explicit how the integrated population distance induced by the projection law controls power.

As a concrete example, we construct a mixed Haar--Fourier projection law. Haar directions emphasize temporally localized changes, whereas Fourier directions emphasize oscillatory and frequency-domain structure. The test combines all sampled directions and is calibrated by permutation. The component labels are retained for a separate descriptive stage based on the directions producing the largest projected discrepancies.

The Haar--Fourier construction is examined through localized, periodic, and combined mean perturbations. The specialized components dominate under their matched alternatives, while the mixture protects against selecting an unsuitable geometry and performs best when both structures are present. An application to ECG5000 indicates that the difference between classes 3 and 4 is predominantly localized in time.

As an additional investigation, we consider a data-adaptive component constructed from functional principal component analysis (FPCA) of the pooled sample. Numerical experiments evaluate whether this adaptive geometry complements the fixed Haar and Fourier components, including under alternatives that differ through second-order structure.

The contributions of the article are therefore fourfold. First, we extend Gaussian separation from a single projection law to arbitrary mixtures of non-degenerate Gaussian measures. Second, we connect the geometry of the projection law to power through the corresponding integrated projected distance. Third, we develop and interpret a mixed Haar--Fourier permutation test. Fourth, we investigate a pooled FPCA component as an adaptive extension.

The remainder of the article is organized as follows. Section~\ref{sec:theory} introduces projected distances and Gaussian mixtures. Section~\ref{sec:power} discusses the dependence of power on the projection law. Section~\ref{sec:hf-law} defines the Haar--Fourier construction. Section~\ref{sec:method} presents the permutation test and descriptive diagnostics. Sections~\ref{sec:simulations} and~\ref{sec:application} report the Haar--Fourier simulations and ECG5000 application. Section~\ref{sec:fpca} presents the pooled FPCA extension and its numerical investigation. Section~\ref{sec:discussion} concludes.

\section{Projection distances and Gaussian mixtures}
\label{sec:theory}

Let $F=L^2([0,1])$ be a real separable Hilbert space with Borel $\sigma$-field $\mathcal B(F)$. For a Borel probability measure $P$ on $F$ and a direction $h\in F$, define the projection map $\pi_h(x)=\langle x,h\rangle$, the projected law $P_h=P\circ\pi_h^{-1}$, and its cumulative distribution function
\[
F_{P,h}(z)=P_h(({-}\infty,z])
=P\{x\in F:\langle x,h\rangle\leq z\},
\qquad z\in\R.
\]
For two probability measures $P$ and $Q$, their Kolmogorov--Smirnov separation in direction $h$ is
\begin{equation}
\label{eq:population-directional-ks}
d_h(P,Q)=\sup_{z\in\R}|F_{P,h}(z)-F_{Q,h}(z)|.
\end{equation}
The projected laws coincide exactly on the set
\[
E(P,Q)=\{h\in F:P_h=Q_h\}
      =\{h\in F:d_h(P,Q)=0\}.
\]

\begin{definition}[Moments condition]
\label{def:moments}
For a Borel probability measure $P$ on $F$, let
\[
m_k(P)=\int_F\|x\|^k\,P(\dd x),\qquad k\geq1.
\]
We say that $P$ satisfies the Moments condition if every $m_k(P)$ is finite and
\[
\sum_{k=1}^{\infty}m_k(P)^{-1/k}=\infty.
\]
\end{definition}

\begin{definition}[Projection continuity]
\label{def:projection-continuity}
A Borel probability measure $P$ on $F$ is projection-continuous if $P_h$ is a continuous probability measure on $\R$ for every $h\in F\setminus\{0\}$.
\end{definition}

For any probability measure $\mu$ on $F$, let $\mu_h=\mu\circ\pi_h^{-1}$ denote its projected law in direction $h$. The measure $\mu$ is Gaussian if $\mu_h$ is a Gaussian distribution on $\R$ for every nonzero $h\in F$. It is non-degenerate if $\mu_h$ has positive variance for every nonzero $h$. Throughout, $I_d$ denotes the $d\times d$ identity matrix.

Let $(\Theta,\mathcal T,\nu)$ be a probability space and let $\{\mu_\theta:\theta\in\Theta\}$ be a measurable family of non-degenerate Gaussian measures on $F$. Their mixture is
\begin{equation}
\label{eq:general-mixture}
W(A)=\int_\Theta\mu_\theta(A)\,\nu(\dd\theta),
\qquad A\in\mathcal B(F).
\end{equation}
A finite mixture is obtained by taking $W=\sum_{j=1}^J\alpha_j\mu_j$, where $J\geq2$, $\alpha_j>0$, and $\sum_{j=1}^J\alpha_j=1$.

The following theorem is the sharp infinite-dimensional Cram\'er--Wold result on which the method is based.

\begin{theorem}[Cuesta-Albertos, Fraiman and Ransford \cite{cuesta2007}]
\label{thm:sharp-cramer-wold}
Let $\mu$ be a non-degenerate Gaussian measure on $F$, and let $P$ and $Q$ be Borel probability measures on $F$. If $P$ satisfies the Moments condition and $\mu(E(P,Q))>0$, then $P=Q$.
\end{theorem}

The next result shows that replacing one Gaussian projection law by a mixture does not weaken the separation property.

\begin{theorem}[Separation by Gaussian mixtures]
\label{thm:mixture-separation}
Let $W$ be the mixture in~\eqref{eq:general-mixture}. Suppose that $P$ satisfies the Moments condition. Then
\[
W(E(P,Q))>0\quad\Longrightarrow\quad P=Q.
\]
Equivalently, if $P\neq Q$, then $d_h(P,Q)>0$ for $W$-almost every $h$.
\end{theorem}

\begin{proof}
By the definition of the mixture, 
\[
W(E(P,Q))
=\int_\Theta\mu_\theta(E(P,Q))\,\nu(\dd\theta).
\]
If this integral is positive, then $\mu_\theta(E(P,Q))>0$ on a set of positive $\nu$-measure, and in particular for at least one $\theta$. Theorem~\ref{thm:sharp-cramer-wold} then implies $P=Q$.
\end{proof}

For a fixed projection law $W$, define the integrated projection distance
\begin{equation}
\label{eq:integrated-distance}
D_W(P,Q)=\int_F d_h(P,Q)\,W(\dd h).
\end{equation}

\begin{proposition}
\label{prop:mixture-distance}
On the class of Borel probability measures satisfying the Moments condition, $D_W$ is a metric.
\end{proposition}

\begin{proof}
The integrand lies in $[0,1]$, so the integral is finite. Symmetry follows immediately. For any $P,Q,S$, the triangle inequality for the supremum norm gives
\[
d_h(P,Q)\leq d_h(P,S)+d_h(S,Q)
\]
for every $h$, and integration yields the triangle inequality for $D_W$. If $P=Q$, then $d_h(P,Q)=0$ for every $h$. Conversely, if $D_W(P,Q)=0$, the nonnegative integrand vanishes $W$-almost surely, so $W(E(P,Q))=1$. Theorem~\ref{thm:mixture-separation} gives $P=Q$.
\end{proof}

\begin{remark}
The mixture measure $W$ in~\eqref{eq:general-mixture} is not necessarily Gaussian. Gaussianity is required conditionally on the mixture components.
\end{remark}

\subsection{Empirical projected distances}
\label{subsec:empirical-distances}

Let $X_1,\ldots,X_n$ and $Y_1,\ldots,Y_m$ be independent samples from $P$ and $Q$, respectively. For $h\in F$, define
\[
\widehat F_{n,h}(z)=\frac1n\sum_{i=1}^n\mathbf 1\{\langle X_i,h\rangle\leq z\},
\qquad
\widehat G_{m,h}(z)=\frac1m\sum_{j=1}^m\mathbf 1\{\langle Y_j,h\rangle\leq z\},
\]
and the empirical directional KS distance
\begin{equation}
\label{eq:empirical-directional-ks}
\widehat d_{n,m}(h)=\sup_{z\in\R}
\left|\widehat F_{n,h}(z)-\widehat G_{m,h}(z)\right|.
\end{equation}
For directions $B_1,\ldots,B_M$, define their empirical average by
\begin{equation}
\label{eq:empirical-projection-average}
\widehat D_{n,m,M}(P,Q)
=\frac1M\sum_{r=1}^M\widehat d_{n,m}(B_r).
\end{equation}
When the design is balanced with $n=m=N_0$, we use the abbreviation
\[
\widehat D_{N_0,M}(P,Q)
:=\widehat D_{N_0,N_0,M}(P,Q).
\]

\section{Projection geometry and power}
\label{sec:power}

The separation theorem shows that every admissible projection law distinguishes different functional laws almost surely. It does not imply that all choices of $W$ perform equally well at finite sample sizes. The relevant population quantity is $D_W(P,Q)$: a projection law assigning greater mass to informative directions produces a larger integrated separation.

The following concentration inequality, established in \cite{galves2026}, applies unchangedto the mixture setting and makes this dependence explicit. Consider the balanced case $n=m=N_0$, and let $B_1,\ldots,B_M$ be independent directions with common law $W$, sampled independently of the data. The statistic $\widehat D_{N_0,M}(P,Q)$ is the balanced-design abbreviation defined in Section~\ref{subsec:empirical-distances}.

\begin{proposition}[Concentration of the empirical projection distance]
\label{prop:concentration}
Suppose that $P$ and $Q$ satisfy the Moments condition and are projection-continuous in the sense of Definitions~\ref{def:moments} and~\ref{def:projection-continuity}. Then there is a constant $C_\star>0$ such that, for every $\gamma>0$,
\begin{equation}
\label{eq:power-concentration}
\Prob\!\left(
\left|\widehat D_{N_0,M}(P,Q)-D_W(P,Q)\right|\geq\gamma
\right)
\leq
2\left(
e^{-M\gamma^2/2}
+e^{-M\gamma^2/32}
+C_\star e^{-N_0\gamma^2/16}
\right).
\end{equation}
\end{proposition}

The inequality separates two sources of approximation. The terms involving $M$ control Monte Carlo integration over the projection directions, while the term involving $N_0$ controls estimation of the projected distribution functions.

For a deterministic threshold $\gamma_\alpha$ satisfying
\begin{equation}
\label{eq:deterministic-threshold}
2\left(
e^{-M\gamma_\alpha^2/2}
+e^{-M\gamma_\alpha^2/32}
+C_\star e^{-N_0\gamma_\alpha^2/16}
\right)\leq\alpha,
\end{equation}
the test that rejects when $\widehat D_{N_0,M}(P,Q)>\gamma_\alpha$ has level at most $\alpha$. Under an alternative for which $D_W(P,Q)>\gamma_\alpha$, equation~\eqref{eq:power-concentration} gives
\begin{align}
\Prob\!\left(\widehat D_{N_0,M}(P,Q)>\gamma_\alpha\right)
\geq 1-2\bigg\{&e^{-M(D_W(P,Q)-\gamma_\alpha)^2/2}
+e^{-M(D_W(P,Q)-\gamma_\alpha)^2/32}\notag\\
&+C_\star e^{-N_0(D_W(P,Q)-\gamma_\alpha)^2/16}\bigg\}.
\label{eq:power-lower-bound}
\end{align}
In particular, the type-II error has an exponential upper bound whose rate is controlled by $D_W(P,Q)^2$. This formalizes the practical role of covariance design: the projection law does not alter whether distinct laws are separated asymptotically, but it can substantially alter the strength of the finite-sample signal.

For $u\geq0$, define the concentration remainder
\begin{equation}
\label{eq:concentration-remainder}
\mathcal R_{N_0,M}(u)
=2\left(e^{-Mu^2/2}+e^{-Mu^2/32}+C_\star e^{-N_0u^2/16}\right).
\end{equation}

The tests implemented in this article use permutation calibration rather than the deterministic threshold in~\eqref{eq:deterministic-threshold}. To define the corresponding exact permutation critical value, let $\pi$ be a uniformly sampled balanced relabeling of the pooled observations, let $\widehat D_{N_0,M}^{\pi}$ denote the statistic in~\eqref{eq:empirical-projection-average} computed after applying $\pi$ while holding $B_1,\ldots,B_M$ fixed, and set
\[
c^{\mathrm{perm}}_{\alpha,N_0,M}
=\inf\left\{t:\Prob_{\pi}\!\left(\widehat D_{N_0,M}^{\pi}\leq t\mid X_1,\ldots,X_{N_0},Y_1,\ldots,Y_{N_0},B_1,\ldots,B_M\right)\geq1-\alpha\right\}.
\]
Fix any $c<D_W(P,Q)$. Then
\begin{align}
\Prob\!\left(
\widehat D_{N_0,M}(P,Q)
\leq c^{\mathrm{perm}}_{\alpha,N_0,M}
\right)
\leq{}&
\Prob\!\left(\widehat D_{N_0,M}(P,Q)\leq c\right)
+\Prob\!\left(c^{\mathrm{perm}}_{\alpha,N_0,M}>c\right)\notag\\
\leq{}&
\mathcal R_{N_0,M}\!\left(D_W(P,Q)-c\right)
+\Prob\!\left(c^{\mathrm{perm}}_{\alpha,N_0,M}>c\right),
\label{eq:permutation-power-bridge}
\end{align}
where $\mathcal R_{N_0,M}$ is defined in~\eqref{eq:concentration-remainder}. Thus the same population distance controls concentration around the alternative, while permutation calibration determines the random critical value. This argument applies directly when the projection law is fixed independently of the data. The pooled FPCA construction introduced in Section~\ref{sec:fpca} is data-dependent and is therefore treated separately.

\section{A mixed Haar--Fourier projection law}
\label{sec:hf-law}

The mixture result permits components with covariance structures adapted to different forms of functional variation. Haar directions emphasize localized changes in time, whereas Fourier directions emphasize oscillatory changes. The component label is retained for interpretation, while all sampled directions can be combined in one test.

\subsection{Population construction}

The Haar orthonormal basis of $L^2([0,1])$ begins with the constant function $e_1^{\mathrm H}(t)=1$ and the Haar wavelet. For resolution $j\geq0$ and location $k=0,\ldots,2^j-1$, let
\[
\psi_{j,k}(t)=2^{j/2}
\left[
\mathbf 1_{[k2^{-j},(k+1/2)2^{-j})}(t)
-\mathbf 1_{[(k+1/2)2^{-j},(k+1)2^{-j})}(t)
\right].
\]

The remaining basis elements $e_r^{\mathrm H}$, $r\geq2$, are obtained by enumerating the functions $\psi_{j,k}$ first by increasing resolution $j$ and then by increasing location $k$.

The real Fourier basis is
\[
e_1^{\mathrm F}(t)=1,
\qquad
e_{2r}^{\mathrm F}(t)=\sqrt2\cos(2\pi r t),
\qquad
e_{2r+1}^{\mathrm F}(t)=\sqrt2\sin(2\pi r t),
\quad r\geq1.
\]
Choose strictly positive summable eigenvalues $(\lambda_r)_{r\geq1}$. For $b\in\{\mathrm H,\mathrm F\}$, define
\[
C_bx=\sum_{r\geq1}\lambda_r
\langle x,e_r^b\rangle e_r^b,
\qquad W_b=\mathcal N(0,C_b).
\]
Each $W_b$ is non-degenerate because every eigenvalue is positive. For $\rho\in(0,1)$, the mixed projection law is
\begin{equation}
\label{eq:hf-mixture}
W_{\mathrm{HF},\rho}=\rho W_{\mathrm H}+(1-\rho)W_{\mathrm F},
\end{equation}
where the right-hand side denotes a mixture of probability measures.
The numerical work uses $\rho=1/2$.

\subsection{Finite-grid implementation}
\label{subsec:grid-implementation}

The implementation represents each discretized curve on a regular grid as a vector in $\R^p$  and uses the Euclidean inner product as a discrete approximation of the $L^2$ inner product.
 In all Haar--Fourier experiments, $p=128$. Let $B_{\mathrm H},B_{\mathrm F}\in\R^{p\times p}$ denote the orthogonal matrices whose columns are the discrete Haar and Fourier basis vectors, respectively.

The finite-grid spectrum is
\begin{equation}
\label{eq:finite-spectrum}
\lambda_{r,p}=c_p
\begin{cases}
1, & 1\leq r\leq K_{\mathrm{HF}},\\
\varepsilon_{\mathrm{HF}}(r-K_{\mathrm{HF}})^{-2}, & K_{\mathrm{HF}}<r\leq p,
\end{cases}
\qquad
c_p^{-1}=K_{\mathrm{HF}}+\varepsilon_{\mathrm{HF}}\sum_{s=1}^{p-K_{\mathrm{HF}}}s^{-2}.
\end{equation}
We use $K_{\mathrm{HF}}=8$ and $\varepsilon_{\mathrm{HF}}=10^{-3}$. Thus every basis coefficient has positive variance and $\sum_{r=1}^p\lambda_{r,p}=1$.

A direction from component $b$ is generated as follows. Draw $a=(a_1,\ldots,a_p)^\top\sim\mathcal N_p(0,I_p)$ and form
\begin{equation}
\label{eq:direction-generation}
g_b=B_b\operatorname{diag}
(\sqrt{\lambda_{1,p}},\ldots,\sqrt{\lambda_{p,p}})a,
\qquad
h_b=\frac{g_b}{\|g_b\|_2}.
\end{equation}
Since the projected KS distance is invariant under positive rescaling of the projection direction, unit-norm rescaling removes the arbitrary scale of the Gaussian draw.

Rather than draw the component label randomly, the experiments use stratified sampling: $M_{\mathrm H}$ independent directions are drawn from the Haar component and $M_{\mathrm F}$ from the Fourier component. Write
\[
M=M_{\mathrm H}+M_{\mathrm F},
\qquad \rho=M_{\mathrm H}/M.
\]
This gives the same component weighting as~\eqref{eq:hf-mixture} while fixing the number of directions drawn from each basis.

\section{Test statistic, permutation calibration, and diagnostics}
\label{sec:method}

Let $X_1,\ldots,X_n\in\R^p$ and $Y_1,\ldots,Y_m\in\R^p$ be independent samples from $P$ and $Q$. We use the empirical directional KS distance $\widehat d_{n,m}(h)$ defined in~\eqref{eq:empirical-directional-ks}, with $\langle X_i,h\rangle=X_i^\top h$ on the finite grid. Denote the Haar directions by $h_{\mathrm H,1},\ldots,h_{\mathrm H,M_{\mathrm H}}$ and the Fourier directions by $h_{\mathrm F,1},\ldots,h_{\mathrm F,M_{\mathrm F}}$. Define
\begin{align}
T_{\mathrm H}^{\mathrm{obs}}
&=\frac1{M_{\mathrm H}}\sum_{r=1}^{M_{\mathrm H}}
\widehat d_{n,m}(h_{\mathrm H,r}),
\label{eq:haar-stat}\\
T_{\mathrm F}^{\mathrm{obs}}
&=\frac1{M_{\mathrm F}}\sum_{r=1}^{M_{\mathrm F}}
\widehat d_{n,m}(h_{\mathrm F,r}),
\label{eq:fourier-stat}\\
T_{\mathrm{HF}}^{\mathrm{obs}}
&=\rho T_{\mathrm H}^{\mathrm{obs}}
+(1-\rho)T_{\mathrm F}^{\mathrm{obs}}.
\label{eq:hf-stat}
\end{align}
The statistics $T^{obs}_H$ and $T^{obs}_F$ are the mean KS distances over all Haar and Fourier directions, respectively, while $T^{obs}_{HF}$ is their convex combination.

\subsection{Permutation calibration}

The directions are drawn independently of the data and held fixed throughout calibration. Pool the observations as $Z_1,\ldots,Z_N$, where $N=n+m$. For permutation replicate $s=1,\ldots,R$, choose uniformly without replacement a subset $I_s\subset\{1,\ldots,N\}$ of size $n$. Observations indexed by $I_s$ form the permuted first sample, and the remaining observations form the permuted second sample. For a direction $h$, write $\widehat d_{n,m}^{(s)}(h)$ for the empirical directional KS distance in~\eqref{eq:empirical-directional-ks} computed with this permuted split. The component statistics $T_b^{(s)}$ are obtained from $\widehat d_{n,m}^{(s)}$ by the same averages used in~\eqref{eq:haar-stat}--\eqref{eq:hf-stat}.

For $b\in\{\mathrm H,\mathrm F,\mathrm{HF}\}$, the Monte Carlo permutation $p$-value is
\begin{equation}
\label{eq:permutation-pvalue}
\widehat p_b=
\frac{1+\sum_{s=1}^R
\mathbf1\{T_b^{(s)}\geq T_b^{\mathrm{obs}}\}}{R+1}.
\end{equation}
At level $\alpha$, the corresponding procedure rejects when $\widehat p_b\leq\alpha$. The principal test uses $\widehat p_{\mathrm{HF}}$; the componentwise tests are retained to show the effect of projection geometry. Conditional on the sampled directions, exchangeability under $H_0:P=Q$ gives finite-sample permutation validity.

The implementation computes all pooled projections once and sorts each projection once. For every observed or permuted label vector, the KS statistic is then obtained from cumulative counts along these fixed orders. This is computationally equivalent to recomputing the empirical CDFs for each permutation.

\begin{proposition}[Consistency of the permutation test]
\label{prop:test-consistency}
Suppose that $P\neq Q$, that $P$ satisfies the Moments condition, and that
\[
\frac{n}{n+m}\longrightarrow\eta\in(0,1).
\]
Let $M_{\mathrm H}$ and $M_{\mathrm F}$ be fixed positive integers. Sample the Haar and Fourier directions independently of the data from their respective non-degenerate Gaussian laws, rescale them to unit norm, and hold them fixed during permutation calibration.

Define
\[
\Delta_{\mathrm H}=\frac1{M_{\mathrm H}}\sum_{r=1}^{M_{\mathrm H}}d_{h_{\mathrm H,r}}(P,Q),
\qquad
\Delta_{\mathrm F}=\frac1{M_{\mathrm F}}\sum_{r=1}^{M_{\mathrm F}}d_{h_{\mathrm F,r}}(P,Q),
\]
and $\Delta_{\mathrm{HF}}=\rho\Delta_{\mathrm H}+(1-\rho)\Delta_{\mathrm F}$. Then, for almost every realization of the sampled directions,
\[
T_b^{\mathrm{obs}}\longrightarrow\Delta_b>0,
\qquad b\in\{\mathrm H,\mathrm F,\mathrm{HF}\}.
\] Moreover, for every fixed permutation replicate $s$,
\[
T_b^{(s)}\longrightarrow0
\quad\text{in probability}.
\]
Consequently, for a fixed number $R$ of Monte Carlo permutations,
\[
\widehat p_b\longrightarrow\frac1{R+1}
\quad\text{in probability},
\]
and the test that rejects when $\widehat p_b\leq\alpha$ is consistent whenever $(R+1)^{-1}\leq\alpha$.
\end{proposition}

\begin{proof}
By Theorem~\ref{thm:mixture-separation}, if $P\neq Q$, the population directional KS distance is positive for almost every direction sampled from either Gaussian component. Multiplication of a direction by a nonzero scalar does not alter the KS distance between the projected laws. The Glivenko--Cantelli theorem therefore gives
\[
\widehat d_{n,m}(h)\longrightarrow d_h(P,Q)>0
\]
for every sampled direction, almost surely. Since the number of directions is fixed, each observed average converges to a strictly positive limit.

Under a random permutation, the two groups are complementary samples drawn without replacement from the same pooled empirical distribution. Their projected empirical CDFs converge uniformly to one another, so $\widehat d_{n,m}^{(s)}(h)\to0$ in probability for each fixed sampled direction. Hence $T_b^{(s)}\to0$. Because $R$ is fixed, with probability tending to one no permuted statistic exceeds the observed statistic, and the Monte Carlo $p$-value converges to its minimum possible value, $(R+1)^{-1}$.
\end{proof}

\subsection{Descriptive basis diagnostics}

The diagnostics are descriptive and do not enter the permutation decision. First, define the mean-score contrast, which compares the overall separation achieved by the two bases through their average directional KS scores
\begin{equation}
\label{eq:mean-diagnostic}
A_{\mathrm{mean}}=
\frac{T_{\mathrm H}^{\mathrm{obs}}-T_{\mathrm F}^{\mathrm{obs}}}
{T_{\mathrm H}^{\mathrm{obs}}+T_{\mathrm F}^{\mathrm{obs}}},
\end{equation}
with value zero when the denominator is zero. Positive values indicate that, on average, the sampled Haar directions produce larger directional KS distances than the sampled Fourier directions.

Second, we focus on the upper tail of the directional KS scores within each basis. For $b\in\{\mathrm H,\mathrm F\}$, let $M_b$ denote the number of sampled directions in component $b$ and define $D_{b,r}^{\mathrm{obs}}=\widehat d_{n,m}(h_{b,r})$, $r=1,\ldots,M_b$. For $q\in(0,1)$ let $k_b=\max\{1,\lceil qM_b\rceil\}$. Write $D_{b,(1)}^{\mathrm{obs}}\leq\cdots\leq D_{b,(M_b)}^{\mathrm{obs}}$ for the ordered values and define the upper-tail mean
\[
U_b(q)=\frac1{k_b}
\sum_{r=M_b-k_b+1}^{M_b}D_{b,(r)}^{\mathrm{obs}}.
\]
The upper-tail basis diagnostic is
\begin{equation}
\label{eq:diagnostic}
A_q=\frac{U_{\mathrm H}(q)-U_{\mathrm F}(q)}
{U_{\mathrm H}(q)+U_{\mathrm F}(q)},
\end{equation}
again set to zero if the denominator is zero. The implementation uses $q=0.1$. Positive values indicate stronger separation among the best sampled Haar directions, while negative values indicate stronger separation among the best Fourier directions.

Finally, set $k=\max\{1,\lceil qM\rceil\}$ and select the $k$ largest statistics among all directions. The pooled top-Haar share is the fraction of selected directions generated by the Haar component. This quantity and $A_{\mathrm{mean}}$ are saved in the raw output, whereas the tables report $A_q$.

\subsection{Computational procedure}

Algorithm~\ref{alg:mixed-projection-test} summarizes the Haar--Fourier procedure.

\begin{algorithm}[htbp]
\caption{Mixed Haar--Fourier projection two-sample test}
\label{alg:mixed-projection-test}
\begin{algorithmic}[1]
\Require Data matrices $X\in\R^{n\times p}$ and $Y\in\R^{m\times p}$; numbers of directions $M_{\mathrm H}$ and $M_{\mathrm F}$; number of permutations $R$; diagnostic fraction $q$.
\Ensure Observed statistics, permutation $p$-values, and descriptive diagnostics.
\State Construct $B_{\mathrm H}$, $B_{\mathrm F}$, and the spectrum in~\eqref{eq:finite-spectrum}.
\State Draw $M_{\mathrm H}$ Haar directions and $M_{\mathrm F}$ Fourier directions according to~\eqref{eq:direction-generation}.
\State Pool the observations and project them onto all sampled directions.
\For{each sampled direction $h$}
  \State Compute the observed directional KS statistic and store the ordering of the pooled projections.
\EndFor
\State Compute $T_{\mathrm H}^{\mathrm{obs}}$, $T_{\mathrm F}^{\mathrm{obs}}$, $T_{\mathrm{HF}}^{\mathrm{obs}}$, and the diagnostics.
\For{$s=1,\ldots,R$}
  \State Permute the labels while keeping directions and projection orders fixed.
  \For{each sampled direction $h$}
    \State Compute the permuted directional KS statistic.
  \EndFor
  \State Compute $T_{\mathrm H}^{(s)}$, $T_{\mathrm F}^{(s)}$, and $T_{\mathrm{HF}}^{(s)}$.
\EndFor
\For{$b\in\{\mathrm H,\mathrm F,\mathrm{HF}\}$}
  \State Compute $\widehat p_b$ according to~\eqref{eq:permutation-pvalue}.
\EndFor
\State \Return the observed statistics, permutation $p$-values, and diagnostics.
\end{algorithmic}
\end{algorithm}

\section{Haar--Fourier numerical example}
\label{sec:simulations}

\subsection{Design}

Functions are represented by vectors with $p=128$ coordinates. In each replication,
\[
X_i=\sigma\xi_i,
\qquad
Y_j=\sigma\zeta_j+\delta g,
\qquad i,j=1,\ldots,35,
\]
where $\xi_i$ and $\zeta_j$ are independent $\mathcal N_p(0,I_p)$ vectors, $\sigma=0.65$, and $\delta=0.8$ under every non-null alternative. Let $b_r^{\mathrm H}$ and $b_r^{\mathrm F}$ denote column $r$ of $B_{\mathrm H}$ and $B_{\mathrm F}$. The four scenarios are:
\begin{itemize}[leftmargin=2em]
\item null: $g=0$ and $\delta=0$;
\item localized: $g=b_7^{\mathrm H}$, a unit-norm Haar vector supported on one quarter of the coordinate domain;
\item frequency: $g=b_8^{\mathrm F}$, the unit-norm vector proportional to $\bigl(\cos(8\pi\ell/p)\bigr)_{\ell=0}^{p-1}$;
\item combined: $g=(b_7^{\mathrm H}+b_8^{\mathrm F})/\|b_7^{\mathrm H}+b_8^{\mathrm F}\|_2$.
\end{itemize}

The elementary perturbations are shown in Figure~\ref{fig:simulation}. Each replication draws new projection directions and new permutations. We use $M_{\mathrm H}=M_{\mathrm F}=50$, $R=199$, 150 Monte Carlo replications, and nominal level $\alpha=0.05$. Empirical power is the fraction of replications for which the corresponding permutation $p$-value is at most $0.05$.

\begin{figure}[htbp]
\centering
\begin{subfigure}[t]{0.48\textwidth}
\centering
\safeincludegraphics[width=\textwidth]{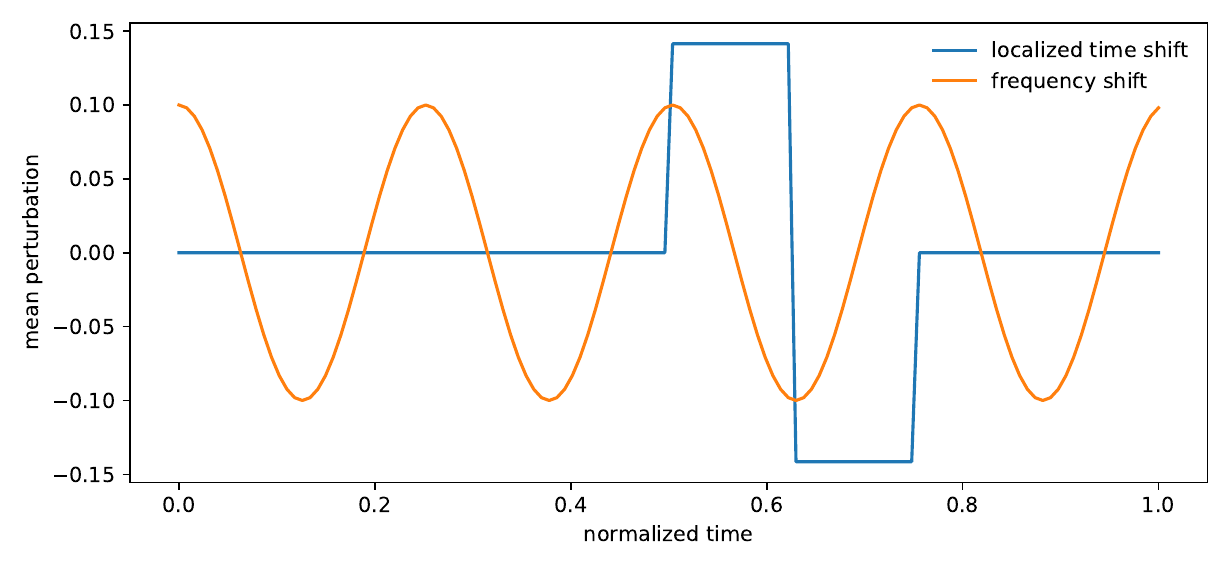}
\caption{Localized and periodic mean perturbations.}
\end{subfigure}
\hfill
\begin{subfigure}[t]{0.48\textwidth}
\centering
\safeincludegraphics[width=\textwidth]{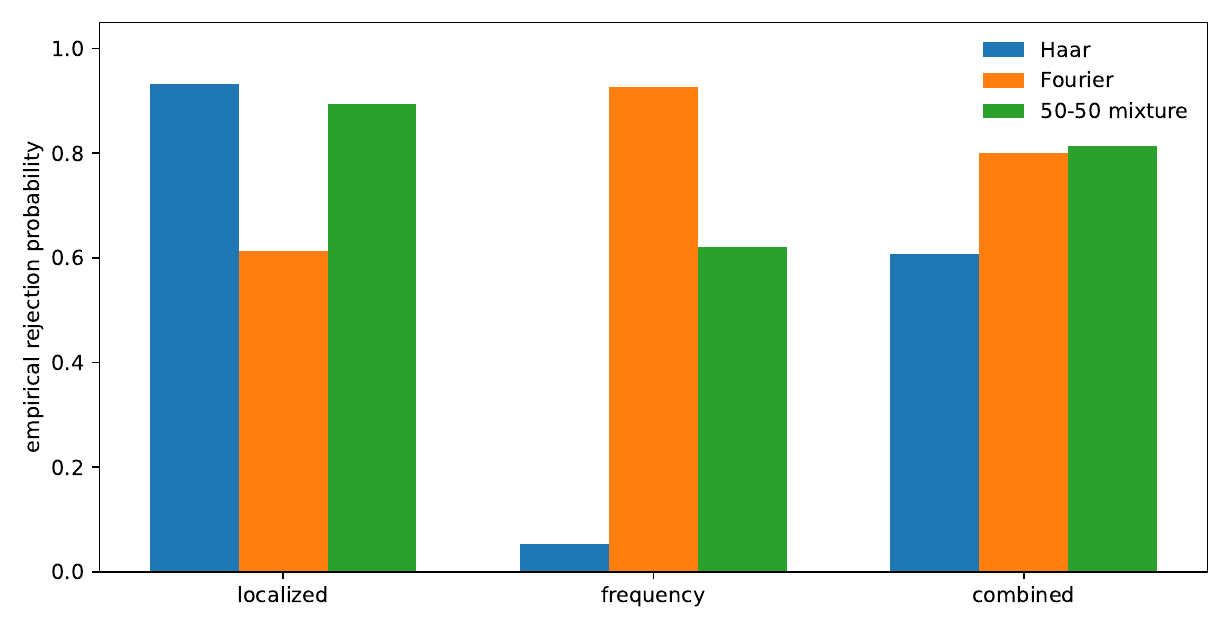}
\caption{Empirical rejection probabilities.}
\end{subfigure}
\caption{Simulation design and power at nominal level $0.05$.}
\label{fig:simulation}
\end{figure}

\subsection{Results}

Table~\ref{tab:simulation} reports empirical rejection probabilities and the median upper-tail diagnostic from~\eqref{eq:diagnostic}. The null rejection rates are close to the nominal level given Monte Carlo uncertainty. Under the localized alternative, the Haar test has power $0.933$, compared with $0.613$ for Fourier, and the mixture retains power $0.893$. Under the frequency alternative, the ordering reverses: Fourier reaches $0.927$, Haar remains near level, and the mixture reaches $0.620$. Under the combined alternative, the mixture has the largest power, $0.813$.

\begin{table}[htbp]
\centering
\small
\begin{tabular}{lrrrr}
\toprule
Scenario & Haar & Fourier & Mixture & Median $A_{0.1}$ \\
\midrule
Null & 0.033 & 0.073 & 0.067 & 0.000 \\
Localized shift & 0.933 & 0.613 & 0.893 & 0.072 \\
Frequency shift & 0.053 & 0.927 & 0.620 & $-0.200$ \\
Combined shift & 0.607 & 0.800 & 0.813 & $-0.042$ \\
\bottomrule
\end{tabular}
\caption{Empirical rejection probabilities over 150 replications and median upper-tail basis diagnostic.}
\label{tab:simulation}
\end{table}

The diagnostic has the intended sign: positive for the localized perturbation, strongly negative for the frequency perturbation, and close to zero for the combined perturbation. The mixture is a hedge rather than a uniformly most powerful choice. It pays a dilution cost when the alternative is purely Fourier, but protects against choosing the wrong geometry and performs best when both geometries are present.

\section{ECG5000 application}
\label{sec:application}

ECG5000 contains 5,000 interpolated heartbeats of length 140 divided into five classes \cite{ecg5000,bagnall2017}. Each heartbeat is linearly interpolated from the 140-point grid $\{\ell/139:\ell=0,\ldots,139\}$ to the 128-point grid $\{\ell/127:\ell=0,\ldots,127\}$.

We compare class 3, with 96 observations, against class 4, with 194 observations. For each $n\in\{10,20,30,40\}$ and each of 100 replications, the procedure independently samples $n$ observations from each class without replacement, draws a new set of directions, and runs the complete test. It uses $M_{\mathrm H}=M_{\mathrm F}=60$, $R=199$, and $\alpha=0.05$. The reported rejection probability is the fraction of balanced subsamples with $\widehat p_b\leq0.05$.

\begin{figure}[htbp]
\centering
\begin{subfigure}[t]{0.48\textwidth}
\centering
\safeincludegraphics[width=\textwidth]{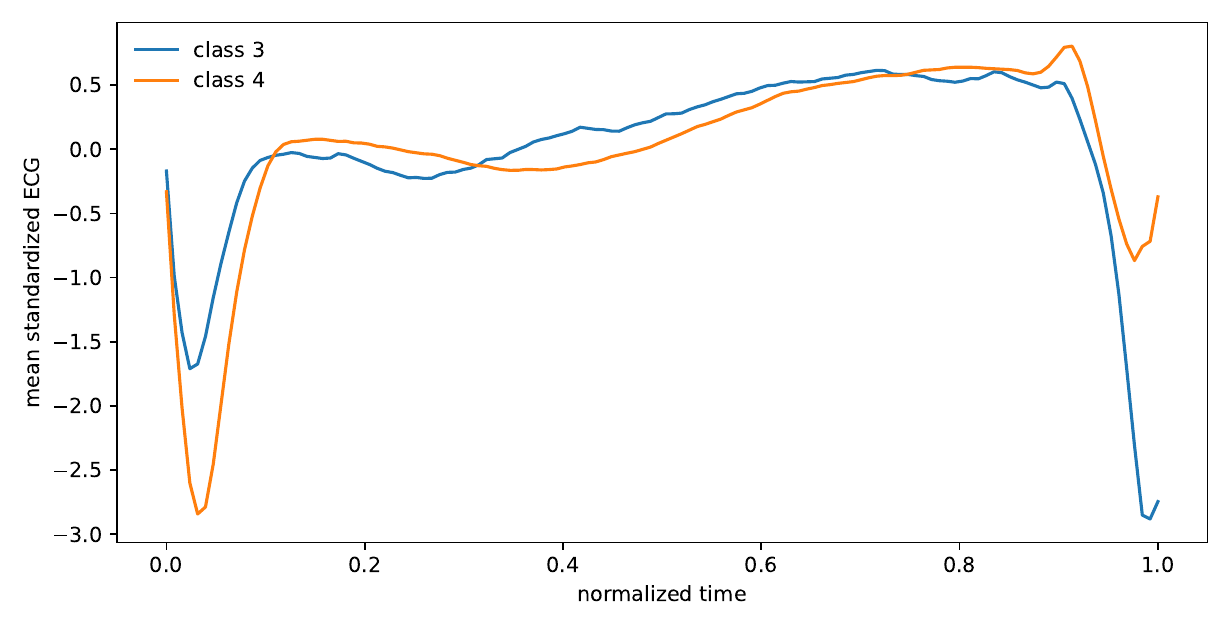}
\caption{Mean standardized heartbeats.}
\end{subfigure}
\hfill
\begin{subfigure}[t]{0.48\textwidth}
\centering
\safeincludegraphics[width=\textwidth]{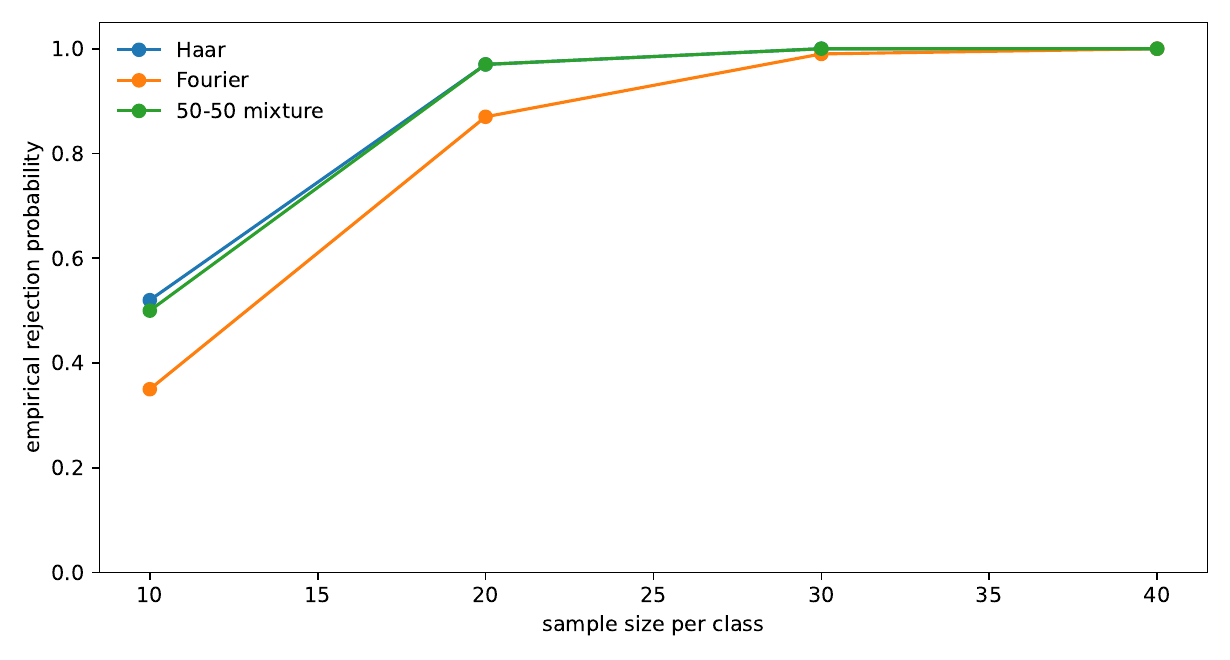}
\caption{Rejection probability versus sample size.}
\end{subfigure}
\caption{ECG5000 classes 3 and 4.}
\label{fig:ecg-power}
\end{figure}

At $n=10$, the Haar component rejects in $52\%$ of the repetitions, the Fourier component in $35\%$, and the mixture in $50\%$. At $n=20$, Haar and the mixture both reach $97\%$, while Fourier reaches $87\%$. All three procedures are essentially saturated by $n=30$.

\begin{table}[htbp]
\centering
\small
\begin{tabular}{rrrrr}
\toprule
$n$ & Haar & Fourier & Mixture & Median $A_{0.1}$ \\
\midrule
10 & 0.52 & 0.35 & 0.50 & 0.033 \\
20 & 0.97 & 0.87 & 0.97 & 0.054 \\
30 & 1.00 & 0.99 & 1.00 & 0.073 \\
40 & 1.00 & 1.00 & 1.00 & 0.092 \\
\bottomrule
\end{tabular}
\caption{ECG5000 class 3 versus class 4: rejection probabilities over 100 balanced subsampling repetitions.}
\label{tab:ecg}
\end{table}

The diagnostic is positive at every sample size and grows with $n$, indicating that the strongest sampled discrepancies are more localized than oscillatory. For descriptive visualization, 1,000 additional unit-norm directions are drawn from each component using the full class samples. Within each component, the direction with the largest observed directional KS statistic is selected. Its sign is chosen so that the projected mean of class 3 does not exceed that of class 4; this sign change leaves the KS statistic unchanged. The best Haar direction attained a projected KS distance of $0.772$, compared with $0.714$ for the best Fourier direction. Figure~\ref{fig:ecg-directions} displays the selected directions and projected empirical distributions. Direction selection is descriptive and is excluded from permutation calibration.

\begin{figure}[htbp]
\centering
\begin{subfigure}[t]{0.48\textwidth}
\centering
\safeincludegraphics[width=\textwidth]{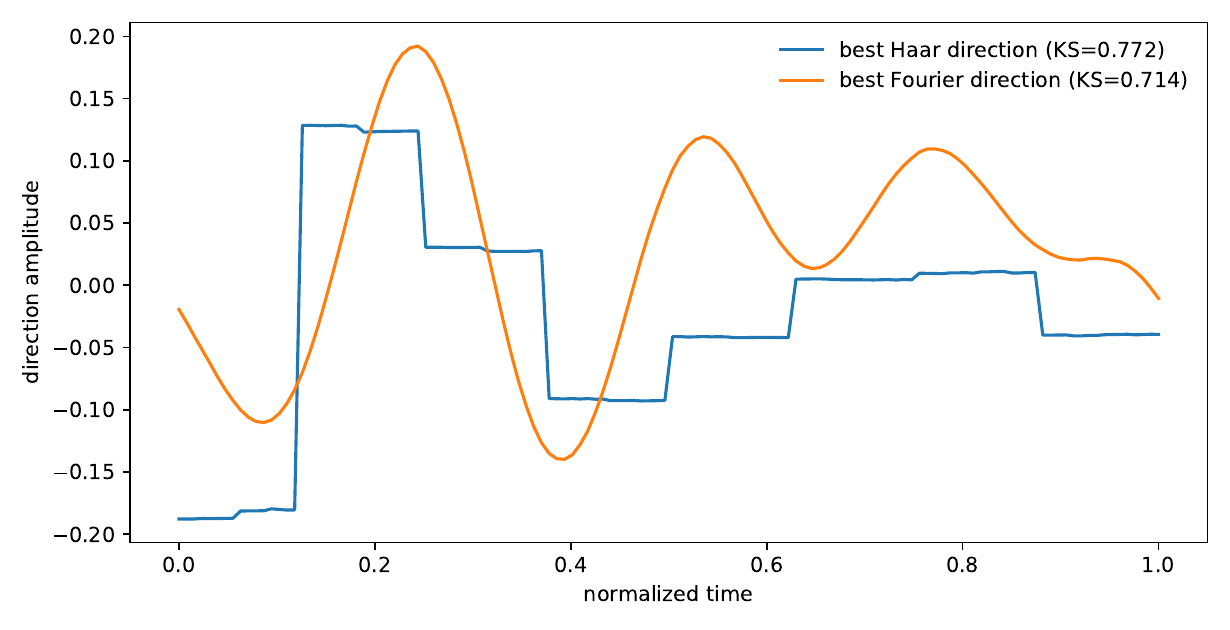}
\caption{Best sampled direction from each component.}
\end{subfigure}
\hfill
\begin{subfigure}[t]{0.48\textwidth}
\centering
\safeincludegraphics[width=\textwidth]{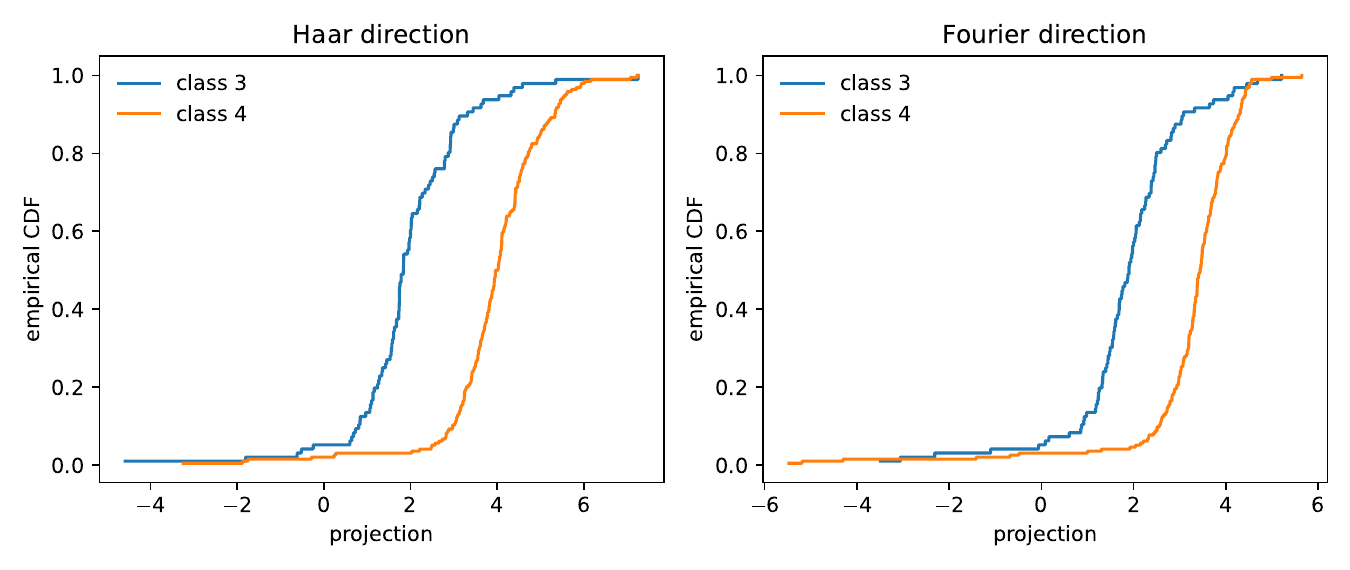}
\caption{Empirical CDFs after projection.}
\end{subfigure}
\caption{Directions that most strongly separate ECG5000 classes 3 and 4. Direction selection uses the full data and is not part of test calibration.}
\label{fig:ecg-directions}
\end{figure}

\section{An adaptive projection component based on pooled FPCA}
\label{sec:fpca}

The preceding sections considered fixed covariance operators selected to emphasize localized or oscillatory structure. We now investigate an additional, data-adaptive projection geometry constructed from the empirical covariance of the pooled sample. The purpose of this extension is to examine whether directions aligned with dominant modes of variation in the observed data can complement the fixed Haar and Fourier geometries.

\subsection{Pooled FPCA construction}
\label{subsec:fpca-construction}

Let
\[
(Z_1,\ldots,Z_N)
=(X_1,\ldots,X_n,Y_1,\ldots,Y_m),
\qquad N=n+m,
\]
denote the pooled sample observed on a common grid $t_1,\ldots,t_p$. Each curve is represented in a cubic B-spline basis with 25 basis functions. The spline coefficients are estimated by penalized least squares using a second-derivative roughness penalty, yielding smooth functional observations \cite{ramsay2005}.

We denote the empirical eigenfunctions
\[
\widehat\phi_1,\ldots,\widehat\phi_{K_{\mathrm P}}
\]
and associated eigenvalues
\[
\widehat\lambda_1,\ldots,\widehat\lambda_{K_{\mathrm P}},
\]
where $K_{\mathrm P}=15$ in all reported simulations.

The empirical eigenfunctions are evaluated on the observation grid and normalized to unit Euclidean norm, producing
\[
B_{\mathrm P}=[b_1^{\mathrm P},\ldots,b_{K_{\mathrm P}}^{\mathrm P}]\in\R^{p\times K_{\mathrm P}},\qquad b_r^{\mathrm P}=\frac{(\widehat\phi_r(t_1),\ldots,\widehat\phi_r(t_p))^\top}{\|(\widehat\phi_r(t_1),\ldots,\widehat\phi_r(t_p))^\top\|_2}.
\]
The eigenvalue weights are regularized and normalized as
\begin{equation}
\label{eq:fpca-eigenvalues}
\lambda_r^{\mathrm P}=
\frac{\widehat\lambda_r+\varepsilon_{\mathrm P}}
{\sum_{k=1}^{K_{\mathrm P}}(\widehat\lambda_k+\varepsilon_{\mathrm P})},
\qquad r=1,\ldots,K_{\mathrm P},
\end{equation}
where $\varepsilon_{\mathrm P}=10^{-8}$. This guarantees positive variance along every retained empirical principal component and has negligible effect on their relative weights.

A conditional FPCA direction is generated by drawing $a\sim\mathcal N_{K_{\mathrm P}}(0,I_{K_{\mathrm P}})$ and setting
\begin{equation}
\label{eq:fpca-direction}
g_{\mathrm P}
=B_{\mathrm P}\operatorname{diag}
(\sqrt{\lambda_1^{\mathrm P}},\ldots,\sqrt{\lambda_{K_{\mathrm P}}^{\mathrm P}})a,
\qquad
h_{\mathrm P}=\frac{g_{\mathrm P}}{\|g_{\mathrm P}\|_2}.
\end{equation}
Conditionally on the pooled sample, $g_{\mathrm P}$ is Gaussian with covariance matrix
\begin{equation}
\label{eq:fpca-covariance}
\widehat C_{\mathrm P}
=\sum_{r=1}^{K_{\mathrm P}}\lambda_r^{\mathrm P} b_r^{\mathrm P}(b_r^{\mathrm P})^\top.
\end{equation}
Because $K_{\mathrm P}<p$, this covariance has rank at most $K_{\mathrm P}$. The resulting projection law is therefore supported on the retained empirical eigenspace and is not non-degenerate on the full ambient space.

The FPCA basis is recomputed from the pooled sample in each simulation replication, allowing the projection geometry to adapt to the observed covariance structure. The construction uses no group labels.

\subsection{Permutation validity}
\label{subsec:fpca-permutation}

The data-dependent construction does not invalidate permutation calibration because it is a symmetric function of the pooled observations. The pooled smoothing, empirical covariance, FPCA basis, eigenvalues, and random directions are computed once and then held fixed throughout the label permutations.

\begin{proposition}[Permutation validity of the pooled FPCA procedure]
\label{prop:fpca-validity}
Assume $H_0:P=Q$. Suppose that the smoothing rule, number of retained components, eigenvalue regularization, and all other FPCA tuning choices are fixed in advance or selected without using the sample labels. Construct the FPCA basis from the pooled sample, sample the projection directions, and hold the complete construction fixed during permutation calibration. Then the permutation test based on the FPCA directions is conditionally valid at its nominal level.
\end{proposition}

\begin{proof}
Under $H_0$, the pooled observations are identically distributed, and conditional on the unordered pooled sample every allocation of $n$ observations to the first group is equally likely. The FPCA construction is invariant to this allocation because it uses only the pooled observations and label-independent auxiliary randomness. Conditional on the pooled sample and the sampled directions, the observed labeling is therefore exchangeable with its permutations. The standard permutation argument gives conditional validity.
\end{proof}

\subsection{Combination with the fixed components}
\label{subsec:hfp-statistic}

Within the three-component procedure, let $M_{\mathrm H}$, $M_{\mathrm F}$, and $M_{\mathrm P}$ denote the numbers of Haar, Fourier, and FPCA directions assigned to that procedure. Compute $T_{\mathrm H}^{\mathrm{obs}}$ and $T_{\mathrm F}^{\mathrm{obs}}$ as in~\eqref{eq:haar-stat} and~\eqref{eq:fourier-stat} using those directions. Draw FPCA directions $h_{\mathrm P,1},\ldots,h_{\mathrm P,M_{\mathrm P}}$ independently from the conditional construction in~\eqref{eq:fpca-direction}, given the pooled sample, and define
\[
T_{\mathrm P}^{\mathrm{obs}}
=\frac1{M_{\mathrm P}}\sum_{r=1}^{M_{\mathrm P}}\widehat d_{n,m}(h_{\mathrm P,r}).
\]
We investigate the equal-weight three-component statistic
\begin{equation}
\label{eq:hfp-statistic}
T_{\mathrm{HFP}}^{\mathrm{obs}}
=\frac13T_{\mathrm H}^{\mathrm{obs}}
+\frac13T_{\mathrm F}^{\mathrm{obs}}
+\frac13T_{\mathrm P}^{\mathrm{obs}}.
\end{equation}
The two-component comparator remains
\[
T_{\mathrm{HF}}^{\mathrm{obs}}
=\frac12T_{\mathrm H}^{\mathrm{obs}}
+\frac12T_{\mathrm F}^{\mathrm{obs}}.
\]
Every reported procedure uses the same total direction budget, and each mixture divides its budget as evenly as integer constraints permit among its active components. The component counts are procedure-specific; the symbols $M_{\mathrm H}$, $M_{\mathrm F}$, and $M_{\mathrm P}$ in~\eqref{eq:hfp-statistic} refer to the three-component run. For permutation replicate $s$, define $T_{\mathrm P}^{(s)}$ from $\widehat d_{n,m}^{(s)}$ and set
\[
T_{\mathrm{HFP}}^{(s)}=\frac13T_{\mathrm H}^{(s)}+\frac13T_{\mathrm F}^{(s)}+\frac13T_{\mathrm P}^{(s)}.
\]
The FPCA-only and H--F--P permutation $p$-values, $\widehat p_{\mathrm P}$ and $\widehat p_{\mathrm{HFP}}$, are defined by~\eqref{eq:permutation-pvalue} with $b=\mathrm P$ and $b=\mathrm{HFP}$, respectively. All component directions and the pooled FPCA construction are held fixed within each permutation run.

The Haar and Fourier components retain the theoretical properties established in Sections~\ref{sec:theory}--\ref{sec:method}. The FPCA component is included as an adaptive empirical direction generator. The theoretical claims for fixed non-degenerate mixtures are not attributed to the finite-rank FPCA component.

\subsection{Additional simulation design}
\label{subsec:fpca-design}

The FPCA experiment first revisits the null, localized, frequency, and combined alternatives from Section~\ref{sec:simulations}, using the same sample-size setting $n=m=35$. Because this is a separate Monte Carlo experiment in which the FPCA basis and all projection directions are regenerated, the empirical rejection probabilities do not coincide exactly with Table~\ref{tab:simulation}.

Two further alternatives compare centered Gaussian processes through their second-order structure. The first is a covariance-length-scale alternative. Both populations have squared exponential covariance kernel
\begin{equation}
\label{eq:sqexp-kernel}
\mathcal K_{\sigma,\ell}(s,t)=\sigma^2
\exp\!\left(-\frac{(s-t)^2}{2\ell^2}\right),\qquad s,t\in[0,1],
\end{equation}
with $\sigma=0.65$. The first population uses $\ell=0.20$, while the second uses $\ell=0.03$. The populations therefore have the same mean function and marginal variance but substantially different smoothness.

The second is an amplitude alternative. Both populations use the same squared exponential correlation structure with $\ell=0.20$, but the first has $\sigma=0.65$ and the second has $\sigma=0.85$. Thus the populations differ in overall variability while sharing their mean and correlation structure.

For these second-order alternatives, the sample sizes are increased to $n=m=60$, because changes in covariance and amplitude were more difficult to detect through one-dimensional projections at the sample size used for the mean-shift experiments. Every row in Table~\ref{tab:power-fpca} is based on 150 Monte Carlo replications and permutation calibration at nominal level $0.05$.

\subsection{Results}
\label{subsec:fpca-results}

Table~\ref{tab:power-fpca} reports empirical rejection probabilities. Under the null, all procedures remain close to the nominal significance level, consistent with the conditional permutation-validity argument in Proposition~\ref{prop:fpca-validity}.

Under the localized alternative, Haar remains the most powerful individual component, reaching $0.927$, compared with $0.533$ for Fourier and $0.733$ for FPCA. The Haar--Fourier mixture attains $0.880$, while the three-component mixture reaches $0.860$. The adaptive component therefore does not remove the dilution cost incurred when a fixed component is closely matched to the alternative.

Under the frequency alternative, Fourier achieves the largest individual power, $0.900$, followed by FPCA at $0.780$, whereas Haar remains close to the nominal level. Adding FPCA improves the combined procedure: the three-component mixture reaches $0.787$, compared with $0.587$ for Haar--Fourier.

The combined alternative illustrates the benefit of complementary geometries. Fourier and FPCA individually attain powers $0.860$ and $0.840$, while the Haar--Fourier--FPCA statistic reaches the largest rejection probability, $0.913$. This exceeds the Haar--Fourier mixture at $0.873$ and each individual component.

Under the covariance-length-scale alternative, FPCA clearly outperforms the fixed components, attaining power $0.880$, while the three-component mixture reaches $0.773$. This behavior is consistent with an adaptive basis aligned with dominant pooled covariance variation. Under the amplitude alternative, all procedures have similar rejection probabilities, between approximately $0.21$ and $0.23$. This departure does not particularly favor any of the projection geometries considered.

\begin{table}[htbp]
\caption{Empirical rejection probabilities over 150 replications. H--F--P denotes the mixture of Haar, Fourier, and FPCA components; H--F denotes the Haar--Fourier mixture. The first four rows use $n=m=35$ and the final two rows use $n=m=60$.}
\label{tab:power-fpca}
\centering
\small
\begin{tabular}{lccccc}
\toprule
& Haar & Fourier & FPCA & H--F--P & H--F \\
\midrule
Null & 0.040 & 0.040 & 0.047 & 0.020 & 0.027 \\
Localized shift & 0.927 & 0.533 & 0.733 & 0.860 & 0.880 \\
Frequency shift & 0.073 & 0.900 & 0.780 & 0.787 & 0.587 \\
Combined shift & 0.680 & 0.860 & 0.840 & 0.913 & 0.873 \\
Covariance length scale & 0.560 & 0.533 & 0.880 & 0.773 & 0.587 \\
Amplitude & 0.213 & 0.227 & 0.207 & 0.233 & 0.220 \\
\bottomrule
\end{tabular}
\end{table}

\section{Discussion}
\label{sec:discussion}

The theoretical result is broader than the Haar--Fourier example. Any probability mixture supported on non-degenerate Gaussian measures preserves separation of functional laws. This permits mixtures over bases, regularity parameters, localization scales, or covariance templates without changing the central separation argument.

The power analysis clarifies why this broad asymptotic guarantee does not make the choice of projection law irrelevant. Finite-sample concentration occurs around $D_W(P,Q)$, and the exponential power bound depends on the size of this integrated distance. Covariance operators that align the sampled directions with a discrepancy can therefore provide substantially larger power even though every admissible fixed Gaussian law is consistent.

The Haar--Fourier experiments illustrate the resulting specialization--robustness tradeoff. A component matched to a pure alternative can outperform the mixture because the mixture allocates some directions to a less informative geometry. The benefit of mixing is protection against severe mismatch and improved performance when multiple structures are present. The mixing weight should encode prior uncertainty: equal weighting is natural when localized and oscillatory departures are similarly plausible, while application-specific weights may be preferable when substantive information is available.

The pooled FPCA investigation  estimates a dominant variation subspace from the pooled observations. This can improve power when relevant differences are represented in the empirical covariance structure, as in the covariance-length-scale experiment. It can also dilute a well-matched fixed component, and it offers no universal advantage under the amplitude alternative. The finite-rank construction is degenerate in the ambient functional space, so it should not be placed under the non-degenerate mixture theorem without an additional argument. Its current justification consists of label-invariant permutation validity and empirical performance.

The ECG5000 application demonstrates how component labels and selected directions can aid interpretation without suggesting that the adaptive FPCA extension was evaluated on those data.

Several extensions remain possible. Fixed mixtures may include spline, wavelet-packet, or other covariance templates, and their regularity parameters may themselves be randomized. Adaptive weights could be estimated on an independent training split. For the FPCA procedure, a separate theoretical analysis could study conditions under which estimated eigenspaces preserve separation, possibly through sample splitting or an explicit non-degenerate regularization. 

\section*{Reproducibility}

The accompanying Haar--Fourier implementation constructs the finite-grid bases, samples directions, computes statistics and diagnostics, and performs permutation calibration. With the ECG5000 files placed in the \texttt{data/} directory, the original numerical outputs are generated by
\begin{verbatim}
python analysis_haar_fourier.py \
  --out results_final \
  --sim-reps 150 \
  --real-reps 100 \
  --permutations 199 \
  --seed 20260720 \
  --data-dir data
\end{verbatim}
The FPCA simulation code additionally smooths each pooled sample, computes the empirical principal components, generates the adaptive directions in~\eqref{eq:fpca-direction}, and applies the same label-permutation principle.

\end{document}